\documentclass[conference]{IEEEtran}
\usepackage{amsmath}
\usepackage{amssymb}
\usepackage{mathrsfs}
\usepackage{cite}
\usepackage{epsfig}
\usepackage{epsfig}
\usepackage{theorem}
\usepackage{graphics}
\usepackage{hyperref}
\usepackage{epsfig}

\newtheorem{thm}{Theorem}

\newtheorem{defi}{Definition}

\begin{document}

\title{The Secrecy Capacity Region of the Degraded Vector Gaussian Broadcast Channel}

\author{Ghadamali Bagherikaram, Abolfazl S. Motahari, Amir K. Khandani\\
Coding and Signal Transmission Laboratory,\\
 Department of
Electrical
and Computer Engineering,\\
 University of Waterloo, Waterloo, Ontario,
 N2L 3G1\\
 Emails: \{gbagheri,abolfazl,khandani\}@cst.uwaterloo.ca
}
 \maketitle
\footnote{Financial support provided by Nortel and the corresponding
matching funds by the Natural Sciences and Engineering Research
Council of Canada (NSERC), and Ontario Centers of Excellence (OCE)
are gratefully acknowledged.}
\begin{abstract}
In this paper, we consider a scenario where a source node wishes to
broadcast two confidential messages for two respective receivers via
a Gaussian MIMO broadcast channel. A wire-tapper also receives the
transmitted signal via another MIMO channel. It is assumed that the
channels are degraded and the wire-tapper has the worst channel. We
establish the capacity region of this scenario. Our achievability
scheme is a combination of the superposition of Gaussian codes and
randomization within the layers which we will refer to as Secret
Superposition Coding. For the outerbound, we use the notion of
enhanced channel to show that the secret superposition of Gaussian
codes is optimal. It is shown that we only need to enhance the
channels of the legitimate receivers, and the channel of the
eavesdropper remains unchanged.
\end{abstract}
\section{Introduction}
Recently there has been significant research conducted in both
theoretical and practical aspects of wireless communication systems
with Multiple-Input Multiple-Output (MIMO) antennas. Most works have
focused on the role of MIMO in enhancing the throughput and
robustness. In this work, however, we focus on the role of such
multiple antennas in enhancing wireless security.

The information-theoretic single user secure communication problem
was first characterized by Wyner in \cite{1}. Wyner considered a
scenario in which a wire-tapper receives the transmitted signal over
a degraded channel with respect to the legitimate receiver's
channel. He measured the level of ignorance at the eavesdropper by
its equivocation and characterized the capacity-equivocation region.
Wyner's work was then extended to the general broadcast channel with
confidential messages by Csiszar et al. \cite{2}. They considered
transmitting confidential information to the legitimate receiver
while transmitting common information to both the legitimate
receiver and the wire-tapper. They established a
capacity-equivocation region of this channel. The secrecy capacity
for the Gaussian wire-tap channel was characterized by
Leung-Yan-Cheong in \cite{3}.

The Gaussian MIMO wire-tap channel has recently been considered by
Khisti et al. in \cite{4,5}. Finding the optimal distribution, which
maximizes the secrecy capacity for this channel is a nonconvex
problem. Khisti et al., however, followed an indirect approach to
evaluate the secrecy capacity of  Csiszar et al. They used a
genie-aided upper bound and characterized the secrecy capacity as
the saddle-value of a min-max problem to show that Gaussian
distribution is optimal. Motivated by the broadcast nature of the
wireless communication systems, we considered the secure broadcast
channel in \cite{6}. In this work, we characterized the secrecy
capacity region of the degraded broadcast channel and showed that
the secret superposition coding is optimal.

The capacity region of the conventional Gaussian MIMO broadcast
channel is studied in \cite{7} by Weingarten et al. The notion of an
enhanced broadcast channel is introduced in this work and is used
jointly with entropy power inequality to characterize the capacity
region of the degraded vector Gaussian broadcast channel. They
showed that the superposition of Gaussian codes is optimal for the
degraded vector Gaussian broadcast channel and that dirty-paper
coding is optimal for the nondegraded case.

In this paper, we aim to characterize the secrecy capacity region of
a secure degraded vector Gaussian MIMO broadcast channel. Our
achievability scheme is a combination of the superposition of
Gaussian codes and randomization within the layers. To prove the
converse, we use the notion of enhanced channel and show that the
secret superposition of Gaussian codes is optimal. We have extended
the results of this paper to the general Gaussian MIMO broadcast
channel in \cite{8} and showed that secret dirty paper coding of
Gaussian codes is optimal.

We acknowledge two other independent and concurrent works of
\cite{9,10} where the authors considered the secrecy capacity region
of the Gaussian MIMO broadcast channel.

The rest of the paper is organized as follows. In section II we
introduce some preliminaries. In section III, we establish the
secrecy capacity region of the Gaussian vector broadcast channel. In
Section V, we conclude the paper.

\section{Preliminaries}
Consider a Secure Gaussian Multiple-Input Multiple-Output Broadcast
Channel (SGMBC) as depicted in Fig. \ref{f1}.
\begin{figure}
\centerline{\includegraphics[scale=.5]{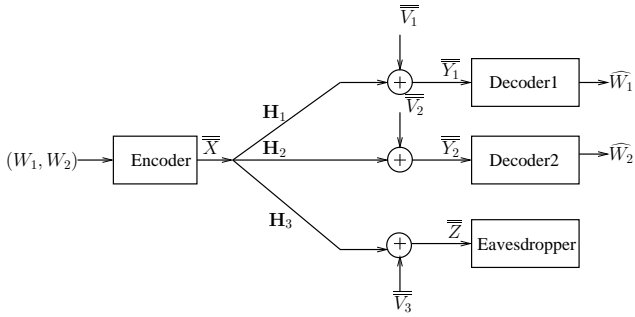}} \caption{Secure Gaussian MIMO Broadcast
Channel} \label{f1}
\end{figure}
In this confidential setting, the transmitter wishes to send two
independent messages $(W_{1},W_{2})$ to the respective receivers in
$n$ uses of the channel and prevent the eavesdropper from having any
information about the messages. At a specific time, the signals
received by the destinations and the eavesdropper are given by
\begin{IEEEeqnarray}{lr}\label{eq1}\nonumber
\mathbf{y_{1}}=\mathbf{H}_{1}\mathbf{x}+\mathbf{n_{1}},\\
\mathbf{y_{2}}=\mathbf{H}_{2}\mathbf{x}+\mathbf{n_{2}},\\
\mathbf{z}=\mathbf{H}_{3}\mathbf{x}+\mathbf{n_{3}}, \nonumber
\end{IEEEeqnarray}
where
\begin{itemize}
\item $\mathbf{x}$ is a real input vector of size $t\times 1$ under an input
covariance constraint. We require that
$E[\mathbf{x}^{T}\mathbf{x}]\preceq \mathbf{S}$ for a positive
semi-definite matrix $\mathbf{S}\succeq 0$.
Here,$\prec,\preceq,\succ$, and $\succeq$ represent partial ordering
between symmetric matrices where $\mathbf{B}\succeq \mathbf{A}$
means that $(\mathbf{B}-\mathbf{A})$ is a positive semi-definite
matrix.
\item $\mathbf{y_{1}}$, $\mathbf{y_{2}}$, and $\mathbf{z}$ are real output
vectors  which are received by the destinations and the eavesdropper
respectively. These are vectors of size $r_{1} \times 1$, $r_{2}
\times 1$, and $r_{3} \times 1$, respectively.
\item $\mathbf{H}_{1}$, $\mathbf{H}_{2}$, and $\mathbf{H}_{3}$ are fixed, real gain matrices
which model the channel gains between the transmitter and the
receivers. These are matrices of size $r_{1} \times t$, $r_{2}
\times t$, and $r_{3} \times t$ respectively. The channel state
information is assumed to be known perfectly at the transmitter and
at all receivers.
\item $\mathbf{n_{1}}$, $\mathbf{n_{2}}$ and $\mathbf{n_{3}}$ are real Gaussian
random vectors with zero means and covariance matrices
$\mathbf{N_{1}}=E[\mathbf{n_{1}}\mathbf{n_{1}}^{T}]\succ 0$,
$\mathbf{N_{2}}=E[\mathbf{n_{2}}\mathbf{n_{2}}^{T}]\succ 0$, and
$\mathbf{N_{3}}=E[\mathbf{n_{3}}\mathbf{n_{3}}^{T}]\succ 0$
respectively.
\end{itemize}
Let $W_{1}$ and $W_{2}$ denote the the message indices of user $1$
and user $2$, respectively. Furthermore, let
$\overline{\overline{X}}$, $\overline{\overline{Y}}_{1}$,
$\overline{\overline{Y}}_{2}$, and $\overline{\overline{Z}}$ denote
the random channel input and random channel outputs matrices over a
block of $n$ samples. Let $\overline{\overline{V}}_{1}$,
$\overline{\overline{V}}_{2}$, and $\overline{\overline{V}}_{3}$
denote the additive noises of the channels. Thus,
\begin{IEEEeqnarray}{lr}\nonumber
\overline{\overline{Y}}_{1}=\mathbf{H}_{1}\overline{\overline{X}}+\overline{\overline{V}}_{1},\\
\overline{\overline{Y}}_{2}=\mathbf{H}_{2}\overline{\overline{X}}+\overline{\overline{V}}_{2},\\ \nonumber
\overline{\overline{Z}}=\mathbf{H}_{3}\overline{\overline{X}}+\overline{\overline{V}}_{3}.
\end{IEEEeqnarray}
Note that $\overline{\overline{V}}_{i}$ is an $r_{i} \times n$
random matrix and $\mathbf{H}_{i}$ is an $r_{i} \times t$
deterministic matrix where $i=1,2,3$. The columns of
$\overline{\overline{V}}_{i}$ are independent Gaussian random
vectors with covariance matrices $\mathbf{N_{i}}$ for $i=1,2,3$. In
addition $\overline{\overline{V}}_{i}$ is independent of
$\overline{\overline{X}}$, $W_{1}$ and $W_{2}$. A
$((2^{nR_{1}},2^{nR_{2}}),n)$ code for the above channel consists of
a stochastic encoder
\begin{equation}
f:(\{1,2,...,2^{nR_{1}}\}\times\{1,2,...,2^{nR_{2}}\})\rightarrow
\overline{\overline{\mathcal{X}}},
\end{equation}
and two decoders,
\begin{equation}
g_{1}:\overline{\overline{\mathcal{Y}}}_{1}\rightarrow \{1,2,...,2^{nR_{1}}\},
\end{equation}
and
\begin{equation}
g_{2}:\overline{\overline{\mathcal{Y}}}_{2}\rightarrow \{1,2,...,2^{nR_{2}}\}.
\end{equation}
where a script letter with double overline denotes the finite alphabet of a random vector.
The average probability of error is defined as the probability that
the decoded messages are not equal to the transmitted messages; that
is,
\begin{equation}
P_{e}^{(n)}=P(g_{1}(\overline{\overline{Y}}_{1})\neq W_{1}\cup g_{2}(\overline{\overline{Y}}_{2})\neq
W_{2}).
\end{equation}

The secrecy levels of confidential messages $W_{1}$ and
$W_{2}$ are measured at the eavesdropper in terms of equivocation
rates, which are defined as follows.
\begin{defi}
The equivocation rates $R_{e1}$, $R_{e2}$ and $R_{e12}$
 for the
secure broadcast channel are:
\begin{IEEEeqnarray}{lr}
R_{e1}=\frac{1}{n}H(W_{1}|\overline{\overline{Z}}),\\
\nonumber R_{e2}=\frac{1}{n}H(W_{2}|\overline{\overline{Z}}), \\
\nonumber
R_{e12}=\frac{1}{n}H(W_{1},W_{2}|\overline{\overline{Z}}).
\end{IEEEeqnarray}
\end{defi}
The perfect secrecy rates $R_{1}$ and $R_{2}$ are the amount of
information that can be sent to the legitimate receivers both
reliably \emph{and} confidentially.
\begin{defi}
A secrecy rate pair $(R_{1},R_{2})$ is said to be achievable if for
any $\epsilon>0,\epsilon_{1}>0,\epsilon_{2}>0,\epsilon_{3}>0$, there
exists a sequence of $((2^{nR_{1}},2^{nR_{2}}),n)$ codes, such that
for sufficiently large $n$,
\begin{IEEEeqnarray}{rl}
\label{l0}P_{e}^{(n)}&\leq \epsilon,\\
\label{l1}
 R_{e1}&\geq R_{1}-\epsilon_{1},\\
\label{l2} R_{e2}&\geq R_{2}-\epsilon_{2},\\
\label{l3}
 R_{e12}&\geq R_{1}+R_{2}-\epsilon_{3}.
\end{IEEEeqnarray}
\end{defi}
In the above definition, the first condition concerns the
reliability, while the other conditions guarantee perfect secrecy
for each individual message and both messages as well. The model
presented in (\ref{eq1}) is SGMBC. For lack of space, the SGMBC
cannot be discussed within this paper, and we will only consider a
subclass of this channel here. The special subclass that we will
consider is the Secure Aligned Degraded MIMO Broadcast Channel
(SADBC). The MIMO broadcast channel of (\ref{eq1}) is said to be
aligned if the number of transmit antennas is equal to the number of
receive antennas at each of the users and the eavesdropper
($t=r_{1}=r_{2}=r_{3}$) and the gain matrices are all identity
matrices
$(\mathbf{H}_{1}=\mathbf{H}_{2}=\mathbf{H}_{3}=\mathbf{I})$.
Furthermore, if the additive noise vectors' covariance matrices are
ordered such that $0\prec \mathbf{N_{1}}\preceq
\mathbf{N_{2}}\preceq \mathbf{N_{3}}$, then the channel is SADBC.

\section{The Capacity Region of The SADBC}
In this section, we characterize the capacity region of the SADBC.
In \cite{6}, we considered the degraded broadcast channel with
confidential messages and establish its secrecy capacity region.
\begin{thm}\label{t1}
The capacity region for transmitting independent secret messages
over the degraded broadcast channel is the convex hull of the
closure of all $(R_{1},R_{2})$ satisfying
\begin{IEEEeqnarray}{rl}\label{l7}
    R_{1}&\leq I(X;Y_{1}|U)-I(X;Z|U), \\
    R_{2}&\leq I(U;Y_{2})-I(U;Z).
\end{IEEEeqnarray}
for some joint distribution $P(u)P(x|u)P(y_{1}, y_{2},z|x)$.
\end{thm}
\begin{proof}
Our achievable coding scheme is based on Cover's superposition scheme
and random binning. We refer to this scheme as the Secret Superposition
Scheme. In this scheme, randomization in the first
layer increases the secrecy rate of the second layer. Our converse proof is based
on a combination of the converse proof of the conventional degraded
broadcast channel and Csiszar Lemma. Please see \cite{6} for details.
\end{proof}
Note that finding optimal distribution which characterizes the
boundary points of (\ref{l7}) and   for the Gaussian channels
involves solving a functional, nonconvex optimization problem.
Usually nontrivial techniques and strong inequalities are used to
solve optimization problems of this type. Indeed, for the single
antenna case, we successfully evaluated the capacity expression of
this scheme in \cite{11}. Liu et al. in \cite{12} evaluated the
capacity expression of MIMO wire-tap channel by using the channel
enhancement method. In the following section, we state and prove our
result for the capacity region of SADBC.

First, we define the achievable rate region due to Gaussian codebook
under a covariance matrix constraint $\mathbf{S}\succeq 0$. The
achievability scheme of Theorem \ref{t1} is the secret superposition
of Gaussian codes and successive decoding at the first receiver.
According to the above theorem, for any covariance matrix input
constraint $\mathbf{S}$ and two semi-definite  matrices
$\mathbf{B_{1}}\succeq 0$ and $\mathbf{B_{2}}\succeq 0$ such that
$\mathbf{B_{1}}+\mathbf{B_{2}} \preceq \mathbf{S}$, it is possible
to achieve the following rates,
\begin{IEEEeqnarray}{rl}\nonumber
R_{1}^{G}(\mathbf{B_{1,2}},&\mathbf{N_{1,2,3}})=\\
\nonumber
&\frac{1}{2}\log|\mathbf{N_{1}^{-1}}(\mathbf{B_{1}}+\mathbf{N_{1}})|-\frac{1}{2}\log|\mathbf{N_{3}^{-1}}(\mathbf{B_{1}}+\mathbf{N_{3}})|,\\
\nonumber R_{2}^{G}(\mathbf{B_{1,2}},&\mathbf{N_{1,2,3}})=\\
\nonumber
&\frac{1}{2}\log\frac{|\mathbf{B_{1}}+\mathbf{B_{2}}+\mathbf{N_{2}}|}{|\mathbf{B_{1}}+\mathbf{N_{2}}|}-\frac{1}{2}\log\frac{|\mathbf{B_{1}}+\mathbf{B_{2}}+\mathbf{N_{3}}|}{|\mathbf{B_{1}}+\mathbf{N_{3}}|}.
\end{IEEEeqnarray}
\begin{defi}
Let $\mathbf{S}$ be a positive semi-definite matrix. Then, the
Gaussian rate region of SADBC under a covariance matrix constraint
$\mathbf{S}$ is given by
\begin{IEEEeqnarray}{rl}\nonumber
&\mathcal{R}^{G}(\mathbf{S},\mathbf{N_{1,2,3}})=   \\&\left\{
                                                       \begin{array}{ll}
                                                          \left(R_{1}^{G}(\mathbf{B_{1,2}},\mathbf{N_{1,2,3}}),R_{2}^{G}(\mathbf{B_{1,2}},\mathbf{N_{1,2,3}})\right)|
                                                           \\ \hbox{s.t}~~\mathbf{S}-(\mathbf{B_{1}}+\mathbf{B_{2}})\succeq 0,~\mathbf{B_{k}}\succeq 0, ~ k=1,2
                                                       \end{array}
                                                     \right\}.
\end{IEEEeqnarray}
\end{defi}
We will show that $\mathcal{R}^{G}(\mathbf{S},\mathbf{N_{1,2,3}})$
is the capacity region of the SADBC. Before that, certain
preliminaries need to be addressed.
\begin{defi}
The rate vector $R^{*}=(R_{1},R_{2})$ is said to be an optimal
Gaussian rate vector under the covariance matrix $\mathbf{S}$, if
$R^{*}\in \mathcal{R}^{G}(\mathbf{S},\mathbf{N_{1,2,3}})$ and if
there is no other rate vector $R^{'*}=(R_{1}^{'},R_{2}^{'})\in
\mathcal{R}^{G}(\mathbf{S},\mathbf{N_{1,2,3}})$ such that
$R_{1}^{'}\geq R_{1}$ and $R_{2}^{'}\geq R_{2}$ where at least one
the inequalities is strict. The set of positive semi-definite
matrices $(\mathbf{B_{1}^{*}},\mathbf{B_{2}^{*}})$ such that
$\mathbf{B_{1}^{*}}+\mathbf{B_{2}^{*}}\preceq \mathbf{S}$ is said to
be realizing matrices of an optimal Gaussian rate vector if the rate
vector
$\left(R_{1}^{G}(\mathbf{B_{1,2}^{*}},\mathbf{N_{1,2,3}}),R_{2}^{G}(\mathbf{B_{1,2}^{*}},\mathbf{N_{1,2,3}})\right)$
is an optimal Gaussian rate vector.
\end{defi}
\begin{defi}
A SADBC with noise covariance matrices of
$(\mathbf{N_{1}^{'}},\mathbf{N_{2}^{'}},\mathbf{N_{3}^{'}})$ is an
enhanced version of another SADBC with noise covariance matrices
$(\mathbf{N_{1}},\mathbf{N_{2}},\mathbf{N_{3}})$ if
\begin{equation}
\mathbf{N_{1}^{'}}\preceq \mathbf{N_{1}},~~\mathbf{N_{2}^{'}}\preceq
\mathbf{N_{2}},~~\mathbf{N_{3}^{'}}=
\mathbf{N_{3}},~~\mathbf{N_{1}^{'}}\preceq \mathbf{N_{2}^{'}}.
\end{equation}
\end{defi}
Obviously, the capacity region of the enhanced version contains the capacity region of the original channel. Note that in characterizing the capacity region of the conventional Gaussian MIMO broadcast channel, all channels must be enhanced by reducing the noise covariance matrices. In our scheme, however, we only enhance the channels for the legitimate receivers and the channel of the eavesdropper remains unchanged. This is due to the fact that the capacity region of the enhanced channel must contain the original capacity region. Reducing the noise covariance matrix of the eavesdropper's channel, however, may reduce the secrecy capacity region. The following theorem connects the definitions of the optimal Gaussian rate vector and the enhanced channel.
\begin{thm}\label{t2}
Consider a SADBC with positive definite noise covariance matrices $(\mathbf{N_{1}},\mathbf{N_{2}},\mathbf{N_{3}})$. Let $\mathbf{B_{1}^{*}}$ and $\mathbf{B_{2}^{*}}$ be realizing matrices of an optimal Gaussian rate vector under a transmit covariance matrix constraint $\mathbf{S}\succ 0$. There then exists an enhanced SADBC with noise covariance matrices $(\mathbf{N_{1}^{'}},\mathbf{N_{2}^{'}},\mathbf{N_{3}^{'}})$ that the following properties hold.
\begin{enumerate}
  \item Enhancement:\\ \nonumber
   $\mathbf{N_{1}^{'}}\preceq \mathbf{N_{1}},~~~\mathbf{N_{2}^{'}}\preceq \mathbf{N_{2}},~~~\mathbf{N_{3}^{'}}= \mathbf{N_{3}},~~~\mathbf{N_{1}^{'}}\preceq \mathbf{N_{2}^{'}}$,
  \item Proportionality:\\ \nonumber
  There exists an $\alpha \geq 0$ and a matrix $\mathbf{A}$ such that \\ \nonumber
  $(\mathbf{I}-\mathbf{A})(\mathbf{B_{1}^{*}}+\mathbf{N_{1}^{'}})=\alpha\mathbf{A}(\mathbf{B_{1}^{*}}+\mathbf{N_{3}^{'}})$,
  \item Rate and optimality preservation: \\ \nonumber
  $R_{k}^{G}(\mathbf{B_{1,2}^{*}},\mathbf{N_{1,2,3}})=R_{k}^{G}(\mathbf{B_{1,2}^{*}},\mathbf{N_{1,2,3}^{'}}) ~~~ \forall k=1,2$,  \nonumber furthermore, $\mathbf{B_{1}^{*}}$ and $\mathbf{B_{2}^{*}}$ are realizing matrices of an optimal Gaussian rate vector in the enhanced channel.
\end{enumerate}
\end{thm}
Theorem \ref{t2} states that if there exists the realizing matrices
of the boundary of $\mathcal{R}^{G}(\mathbf{S},\mathbf{N_{1,2,3}})$,
then the secret superposition coding with Gaussian codebook is the
optimal choice for the capacity region of a SADBC. Note that this
Theorem provides a sufficient condition to evaluate the capacity
region of SADBC.

\begin{proof}
The realizing matrices $\mathbf{B_{1}^{*}}$ and $\mathbf{B_{2}^{*}}$
are the solution of the following optimization problem:
\begin{IEEEeqnarray}{rl}
&\max_{(\mathbf{B_{1}},\mathbf{B_{2}})} R_{1}^{G}(\mathbf{B_{1,2}},\mathbf{N_{1,2,3}})+\mu R_{2}^{G}(\mathbf{B_{1,2}},\mathbf{N_{1,2,3}})\\ \nonumber
&\hbox{s.t}~~ \mathbf{B_{1}}\succeq 0,~~~\mathbf{B_{2}}\succeq 0,~~~\mathbf{B_{1}}+\mathbf{B_{2}}\preceq\mathbf{S},
\end{IEEEeqnarray}
where $\mu \geq 1$. Using the Lagrange Multiplier method, the above
constraint optimization problem is equivalent to the following
unconditional optimization problem:
\begin{IEEEeqnarray}{rl}\nonumber
&\max_{(\mathbf{B_{1}},\mathbf{B_{2}})}
R_{1}^{G}(\mathbf{B_{1,2}},\mathbf{N_{1,2,3}})+\mu
R_{2}^{G}(\mathbf{B_{1,2}},\mathbf{N_{1,2,3}})\\
\nonumber&+Tr\{\mathbf{B_{1}}\mathbf{O_{1}}\}
+Tr\{\mathbf{B_{2}}\mathbf{O_{2}}\}+Tr\{\mathbf{(S-B_{1}-B_{2})}\mathbf{O_{3}}\},
\end{IEEEeqnarray}
where $\mathbf{O_{1}}$, $\mathbf{O_{2}}$, and $\mathbf{O_{3}}$ are
positive semi-definite $t \times t$ matrices such that
$Tr\{\mathbf{B_{1}^{*}}\mathbf{O_{1}}\}=0$,
$Tr\{\mathbf{B_{2}^{*}}\mathbf{O_{2}}\}=0$, and
$Tr\{\mathbf{(S-B_{1}^{*}-B_{2}^{*})}\mathbf{O_{3}}\}=0$. As all
$\mathbf{B_{k}^{*}}, ~k=1,2$, $\mathbf{O_{i}},~ i=1,2,3$, and
$\mathbf{S-B_{1}^{*}-B_{2}^{*}}$ are positive semi-definite
matrices, then we must have
$\mathbf{B_{k}^{*}}\mathbf{O_{k}}=0,~~k=1,2$ and
$(\mathbf{S-B_{1}^{*}-B_{2}^{*}})\mathbf{O_{3}}=0$. According to the
necessary KKT conditions, and after some manipulations we have:
\begin{IEEEeqnarray}{rl}\label{eq3}\nonumber
     (\mathbf{B_{1}^{*}}+\mathbf{N_{1}})^{-1}+(\mu-1)(\mathbf{B_{1}^{*}}+\mathbf{N_{3}})^{-1}+\mathbf{O_{1}}&=\\ \mu(\mathbf{B_{1}^{*}}+\mathbf{N_{2}})^{-1}+\mathbf{O_{2}},
     \end{IEEEeqnarray}
     \begin{IEEEeqnarray}{rl}\nonumber
\mu(\mathbf{B_{1}^{*}}+\mathbf{B_{2}^{*}}+\mathbf{N_{2}})^{-1}+\mathbf{O_{2}}&= \mu(\mathbf{B_{1}^{*}}+\mathbf{B_{2}^{*}}+\mathbf{N_{3}})^{-1}\\&+\mathbf{O_{3}}.
\end{IEEEeqnarray}
We choose the noise covariance matrices of the enhanced SADBC as the following:
\begin{IEEEeqnarray}{rl}
\mathbf{N_{1}^{'}}&=\left(\mathbf{N_{1}}^{-1}+\mathbf{O_{1}}\right)^{-1},\\
\nonumber
\mathbf{N_{2}^{'}}&=\left(\left(\mathbf{B_{1}^{*}}+\mathbf{N_{2}}\right)^{-1}+\frac{1}{\mu}\mathbf{O_{2}}\right)^{-1}-\mathbf{B_{1}^{*}},\\
\nonumber \mathbf{N_{3}^{'}}&=\mathbf{N_{3}}.
\end{IEEEeqnarray}
As $\mathbf{O_{1}}\succeq 0$ and $\mathbf{O_{2}}\succeq 0$, then the
above choice has the enhancement property. The expression
$\left((\mathbf{B_{1}^{*}}+\mathbf{N_{1}})^{-1}+\mathbf{O_{1}}\right)^{-1}$
can be written as:
\begin{IEEEeqnarray}{rl}\nonumber
\left((\mathbf{B_{1}^{*}}+\mathbf{N_{1}})^{-1}+\mathbf{O_{1}}\right)^{-1}&=
\big((\mathbf{B_{1}^{*}}+\mathbf{N_{1}})^{-1}\\ \nonumber &\left(\mathbf{I}+(\mathbf{B_{1}^{*}}+\mathbf{N_{1}})\mathbf{O_{1}}\right)\big)^{-1}\\
\nonumber&\stackrel{(a)}{=}\left(\mathbf{I}+\mathbf{N_{1}}\mathbf{O_{1}}\right)^{-1}(\mathbf{B_{1}^{*}}+\mathbf{N_{1}})\\ \nonumber &-\mathbf{B_{1}^{*}}+\mathbf{B_{1}^{*}}\\
\nonumber&=\left(\mathbf{I}+\mathbf{N_{1}}\mathbf{O_{1}}\right)^{-1}\big((\mathbf{B_{1}^{*}}+\mathbf{N_{1}})\\ \nonumber &-\left(\mathbf{I}+\mathbf{N_{1}}\mathbf{O_{1}}\right)\mathbf{B_{1}^{*}}\big)+\mathbf{B_{1}^{*}}\\
\nonumber&\stackrel{(b)}{=}\left(\mathbf{I}+\mathbf{N_{1}}\mathbf{O_{1}}\right)^{-1}\mathbf{N_{1}}+\mathbf{B_{1}^{*}}\\
\nonumber&=\left(\mathbf{N_{1}}\left(\mathbf{N_{1}^{-1}}+\mathbf{O_{1}}\right)\right)^{-1}\mathbf{N_{1}}+\mathbf{B_{1}^{*}}\\
\nonumber&=\left(\mathbf{N_{1}^{-1}}+\mathbf{O_{1}}\right)^{-1}+\mathbf{B_{1}^{*}}\\
\nonumber &=\mathbf{B_{1}^{*}}+\mathbf{N_{1}^{'}},
\end{IEEEeqnarray}
where $(a)$ and $(b)$ follows from the fact that $\mathbf{B_{1}^{*}}\mathbf{O_{1}}=0$. Similarly, it can be shown that
\begin{IEEEeqnarray}{rl}\nonumber
\mu(\mathbf{B_{1}^{*}}+\mathbf{N_{2}})^{-1}+\mathbf{O_{2}}=\mu(\mathbf{B_{1}^{*}}+\mathbf{N_{2}^{'}})^{-1},
\end{IEEEeqnarray}
Therefore, according to (\ref{eq3}) the following property holds for the enhanced channel.
\begin{IEEEeqnarray}{rl}\label{eq4}\nonumber
     (\mathbf{B_{1}^{*}}+\mathbf{N_{1}^{'}})^{-1}+(\mu-1)(\mathbf{B_{1}^{*}}+\mathbf{N_{3}^{'}})^{-1}= \mu(\mathbf{B_{1}^{*}}+\mathbf{N_{2}^{'}})^{-1}.
\end{IEEEeqnarray}
Since $\mathbf{N_{1}^{'}}\preceq \mathbf{N_{2}^{'}}\preceq
\mathbf{N_{3}^{'}}$ then, there exists a matrix $\mathbf{A}$ such
that
$\mathbf{N_{2}^{'}}=(\mathbf{I}-\mathbf{A})\mathbf{N_{1}^{'}}+\mathbf{A}\mathbf{N_{3}^{'}}$
where
$\mathbf{A}=(\mathbf{N_{2}^{'}}-\mathbf{N_{1}^{'}})(\mathbf{N_{3}^{'}}-\mathbf{N_{1}^{'}})^{-1}$.
Therefore, the above equation can be written as.
\begin{IEEEeqnarray}{rl}\nonumber
     &(\mathbf{B_{1}^{*}}+\mathbf{N_{1}^{'}})^{-1}+(\mu-1)(\mathbf{B_{1}^{*}}+\mathbf{N_{3}^{'}})^{-1}=\\ \nonumber&
     \mu\left[(\mathbf{I}-\mathbf{A})(\mathbf{B_{1}^{*}}+\mathbf{N_{1}^{'}})+\mathbf{A}(\mathbf{B_{1}^{*}}+\mathbf{N_{3}^{'}})\right]^{-1}.
\end{IEEEeqnarray}
Let
$(\mathbf{I}-\mathbf{A})(\mathbf{B_{1}^{*}}+\mathbf{N_{1}^{'}})=\alpha\mathbf{A}(\mathbf{B_{1}^{*}}+\mathbf{N_{3}^{'}})$
then after some manipulations, the above equation becomes
\begin{IEEEeqnarray}{rl}
     \frac{1}{\alpha}\mathbf{I}+(\mu-1-\frac{1}{\alpha})\mathbf{A}=\frac{\mu}{\alpha+1}\mathbf{I}.
\end{IEEEeqnarray}
The above equation is satisfied by $\alpha=\frac{1}{\mu-1}$ which completes the proportionality property.

We can now prove the rate conservation property. The expression $\frac{|\mathbf{B_{1}^{*}}+\mathbf{N_{1}^{'}}|}{|\mathbf{N_{1}^{'}}|}$ can be written as follow.
\begin{IEEEeqnarray}{rl}\label{eq5}
\frac{|\mathbf{B_{1}^{*}}+\mathbf{N_{1}^{'}}|}{|\mathbf{N_{1}^{'}}|}&=\frac{|\mathbf{I}|}{|\mathbf{N_{1}^{'}}\left(\mathbf{B_{1}^{*}}+\mathbf{N_{1}^{'}}\right)^{-1}|}\\
\nonumber &=\frac{|\mathbf{I}|}{|\left(\mathbf{B_{1}^{*}}+\mathbf{N_{1}^{'}}-\mathbf{B_{1}^{*}}\right)\left(\mathbf{B_{1}^{*}}+\mathbf{N_{1}^{'}}\right)^{-1}|}\\
\nonumber &=\frac{|\mathbf{I}|}{|\mathbf{I}-\mathbf{B_{1}^{*}}\left(\mathbf{B_{1}^{*}}+\mathbf{N_{1}^{'}}\right)^{-1}|}\\
\nonumber &=\frac{|\mathbf{I}|}{|\mathbf{I}-\mathbf{B_{1}^{*}}\left((\mathbf{B_{1}^{*}}+\mathbf{N_{1}})^{-1}+\mathbf{O_{1}}\right)|}\\
\nonumber &\stackrel{(a)}{=}\frac{|\mathbf{I}|}{|\mathbf{I}-\mathbf{B_{1}^{*}}\left(\mathbf{B_{1}^{*}}+\mathbf{N_{1}}\right)^{-1}|}\\
\nonumber &=\frac{|\mathbf{B_{1}^{*}}+\mathbf{N_{1}}|}{|\mathbf{N_{1}}|},
\end{IEEEeqnarray}
where $(a)$ once again follows from the fact that
$\mathbf{B_{1}^{*}}\mathbf{O_{1}}=0$. To complete the proof of rate
conservation, consider the following equalities.
\begin{IEEEeqnarray}{rl}\label{eq6}
\frac{|\mathbf{B_{1}^{*}}+\mathbf{B_{2}^{*}}+\mathbf{N_{2}^{'}}|}{|\mathbf{B_{1}^{*}}+\mathbf{N_{2}^{'}}|}&=\frac{|\mathbf{B_{2}^{*}}\left(\mathbf{B_{1}^{*}}+\mathbf{N_{2}^{'}}\right)^{-1}+\mathbf{I}|}{|\mathbf{I}|}\\ \nonumber
&=\frac{|\mathbf{B_{2}^{*}}\left((\mathbf{B_{1}^{*}}+\mathbf{N_{2}})^{-1}+\frac{1}{\mu}\mathbf{O_{2}}\right)+\mathbf{I}|}{|\mathbf{I}|}\\
\nonumber
&\stackrel{(a)}{=}\frac{|\mathbf{B_{1}^{*}}+\mathbf{B_{2}^{*}}+\mathbf{N_{2}}|}{|\mathbf{B_{1}^{*}}+\mathbf{N_{2}}|},
\end{IEEEeqnarray}
where\ $(a)$ follows from the fact
$\mathbf{B_{2}^{*}}\mathbf{O_{2}}=0$. Therefore, according to
(\ref{eq5}), (\ref{eq6}), and the fact that
$\mathbf{N_{3}^{'}}=\mathbf{N_{3}}$, the rate preservation property
holds for the enhanced channel. To prove the optimality
preservation, we need to show that
$(\mathbf{B_{1}^{*}},\mathbf{B_{2}^{*}})$ are also realizing
matrices of an optimal Gaussian rate vector in the enhanced channel.
For that purpose, note that the necessary KKT conditions for the
enhanced channel coincides with the KKT conditions of the original
channel.
\end{proof}
We can now use Theorem \ref{t2} to prove that
$\mathcal{R}^{G}(\mathbf{S},\mathbf{N_{1,2,3}})$ is the capacity
region of the SADBC. We follow Bergmans' approach \cite{13} to prove
a contradiction. Note that since the original channel is not
proportional, we cannot apply  Bergmans' proof on the original
channel directly. Here we apply his proof on the enhanced channel
instead.
\begin{thm}
Consider a SADBC with positive definite noise covariance matrices
$(\mathbf{N_{1}},\mathbf{N_{2}},\mathbf{N_{3}})$. Let
$\mathcal{C}(\mathbf{S},\mathbf{N_{1,2,3}})$ denote the capacity
region of the SADBC under a covariance matrix constraint
$\mathbf{S}\succ 0$ . Then,
$\mathcal{C}(\mathbf{S},\mathbf{N_{1,2,3}})=\mathcal{R}^{G}(\mathbf{S},\mathbf{N_{1,2,3}})$.
\end{thm}
\begin{proof}
The achievability scheme is secret superposition coding with
Gaussian codebook. For the converse proof, we use a contradiction
argument and assume that there exists an achievable rate
vector$(R_{1},R_{2})$ which is not in the Gaussian region. We can
apply the steps of Bergmans' proof of \cite{10} on the enhanced
channel to show that this assumption is impossible.
According to the Theorem \ref{t1}, $R_{1}$ is bounded as follows.
\begin{IEEEeqnarray}{rl}\nonumber
R_{1}&\leq h(\mathbf{y_{1}}|\mathbf{u})-h(\mathbf{z}|\mathbf{u})-\left(h(\mathbf{y_{1}}|\mathbf{x},\mathbf{u})-h(\mathbf{z}|\mathbf{x},\mathbf{u})\right)\\
\nonumber&=  h(\mathbf{y_{1}}|\mathbf{u})-h(\mathbf{z}|\mathbf{u})-\frac{1}{2}\left(\log|\mathbf{N_{1}^{'}}|-\log|\mathbf{N_{3}^{'}}|)\right)
\end{IEEEeqnarray}
Since $R_{1}>R_{1}^{G}(\mathbf{B_{1,2}},\mathbf{N_{1,2,3}^{'}})$, the above inequality means that
\begin{IEEEeqnarray}{rl}\nonumber
h(\mathbf{y_{1}}|\mathbf{u})-h(\mathbf{z}|\mathbf{u})>\frac{1}{2}\left(\log|\mathbf{B_{1}^{*}}+\mathbf{N_{1}^{'}}|-\log|\mathbf{B_{1}^{*}}+\mathbf{N_{3}^{'}}|)\right)
\end{IEEEeqnarray}
By the definition of matrix $\mathbf{A}$ and since
$\mathbf{y_{1}}\rightarrow\mathbf{y_{2}}\rightarrow\mathbf{z}$ forms
a Morkov chain, the received signals $\mathbf{z}$ and
$\mathbf{y_{2}}$ can be written as
$\mathbf{z}=\mathbf{y_{1}}+\mathbf{\widetilde{n}}$ and
$\mathbf{y_{2}}=\mathbf{y_{1}}+\mathbf{A}^{\frac{1}{2}}\mathbf{\widetilde{n}}$
where $\mathbf{\widetilde{n}}$ is an independent Gaussian noise with
covariance matrix
$\mathbf{\widetilde{N}}=\mathbf{N_{3}^{'}}-\mathbf{N_{1}^{'}}$.
According to Costa's Entropy Power Inequality and the previous
inequality, we have
\begin{IEEEeqnarray}{rl}\label{eq7}
\nonumber
h(\mathbf{y_{2}}|\mathbf{u})-&h(\mathbf{z}|\mathbf{u})
 \\ \nonumber &\geq\frac{t}{2}\log\left(|\mathbf{I}-\mathbf{A}|^{\frac{1}{t}}2^{\frac{2}{t}\left(h(\mathbf{y_{1}|\mathbf{u}})-h(\mathbf{z|\mathbf{u}})\right)}+|\mathbf{A}|^{\frac{1}{t}})\right)\\
\nonumber &>\frac{t}{2}\log\left(\frac{|\mathbf{I}-\mathbf{A}|^{\frac{1}{t}}|\mathbf{B_{1}^{*}}+\mathbf{N_{1}^{'}}|^{\frac{1}{t}}}{|\mathbf{B_{1}^{*}}+\mathbf{N_{3}^{'}}|^{\frac{1}{t}}}+|\mathbf{A}|^{\frac{1}{t}})\right)\\
&\stackrel{(a)}{=}\frac{1}{2}\log(\mathbf{B_{1}^{*}}+\mathbf{N_{2}^{'}})-\frac{1}{2}\log(\mathbf{B_{1}^{*}}+\mathbf{N_{3}^{'}})
\end{IEEEeqnarray}
where (a) is due to the proportionality property. The rate $R_{2}$ is bounded as follows
\begin{IEEEeqnarray}{rl}\nonumber
R_{2}&\leq h(\mathbf{y_{2}})-h(\mathbf{z})-\left(h(\mathbf{y_{2}}|\mathbf{u})-h(\mathbf{z}|\mathbf{u})\right)
\end{IEEEeqnarray}
Using (\ref{eq7}) and the fact that
$R_{2}>R_{2}^{G}(\mathbf{B_{1,2}},\mathbf{N_{1,2,3}^{'}})$, the
above inequality means that
\begin{IEEEeqnarray}{rl}\nonumber
&h(\mathbf{y_{2}})-h(\mathbf{z})\geq R_{2}+h(\mathbf{y_{2}}|\mathbf{u})-h(\mathbf{z}|\mathbf{u})>\\
\nonumber
&\frac{1}{2}\log(\mathbf{B_{1}^{*}}+\mathbf{B_{2}^{*}}+\mathbf{N_{2}^{'}})-\frac{1}{2}\log(\mathbf{B_{1}^{*}}+\mathbf{B_{2}^{*}}+\mathbf{N_{3}^{'}})
\end{IEEEeqnarray}
which is a contradiction with the fact that Gaussian distribution
maximizes
$h(\mathbf{x}+\mathbf{n_{2}})-h(\mathbf{x}+\mathbf{n_{3}})$
\cite{14}.
\end{proof}
\section{Conclusion}
A scenario where a source node wishes to broadcast two confidential
messages for two respective receivers via a Gaussian MIMO broadcast
channel, while a wire-tapper also receives the transmitted signal
via another MIMO channel is considered. We considered the secure
vector Gaussian degraded  broadcast channel and established its
capacity region. Our achievability scheme is the secret
superposition of Gaussian codes. Instead of solving a nonconvex
problem, we used the notion of an enhanced channel to show that
secret superposition of Gaussian codes is optimal. To characterize
the secrecy capacity region of the vector Gaussian degraded
broadcast channel, we only enhanced the channels for the legitimate
receivers, and the channel of the eavesdropper remains unchanged.

\end{document}